\documentclass[conference]{IEEEtran}
\pdfoutput=1

\usepackage[sorting=none,style=ieee,backend=bibtex]{biblatex}
\usepackage{amsmath}
\usepackage{amsbsy}
\usepackage[noend]{algpseudocode}
\usepackage{graphicx}
\usepackage{mathtools}
\usepackage{algorithm}
\usepackage{amssymb}
\usepackage{eqnarray}
\usepackage{multicol}
\usepackage{listings}
\usepackage{color}
\usepackage{units}
\usepackage{amsthm}
\usepackage{amsmath}
\usepackage{vwcol}
\usepackage{footnote}
\usepackage{blkarray}
\usepackage{comment}
\usepackage{tikz}
\usepackage{textcomp}
\usepackage{hyperref}

\newtheorem{theorem}{Theorem}

\newtheorem{definition}{Definition}

\def\tablename{Table}
\makesavenoteenv{tabular}
\makesavenoteenv{table}

\allowdisplaybreaks

\ifCLASSINFOpdf

\newcommand{\ra}{\rightarrow}
\newcommand{\mc}[1]{\mathcal{#1}}
\newcommand{\calx}{\mc{X}}
\newcommand{\caly}{\mc{Y}}
\newcommand{\calz}{\mc{Z}}

\newcommand{\tabp}[2]{\theta^{#1}_{#2}}

\newcommand{\ccdiyxz}{I(Y^n\ra X^n \cc Z^n)}
\newcommand{\ccdiryxz}{\bar{I}(Y\ra X \cc Z)}
\newcommand{\xm}[1]{x_{i-#1}}
\newcommand{\ym}[1]{y_{i-#1}}
\newcommand{\zm}[1]{z_{i-#1}}

\newcommand{\R}{\mathbb{R}}
\newcommand{\cc}{\mid\mid}

\begin{document}

\title{On the Bias of Directed Information Estimators}

\author{
Gabriel~Schamberg, \emph{Student Member, IEEE},
Todd~P.~Coleman, \emph{Senior Member,~IEEE}
}

\maketitle



\begin{abstract}
When estimating the directed information between two jointly stationary Markov processes, it is typically assumed that the recipient of the directed information is itself Markov of the same order as the joint process. While this assumption is often made explicit in the presentation of such estimators, a characterization of when we can expect the assumption to hold is lacking. Using the concept of d-separation from Bayesian networks, we present sufficient conditions for which this assumption holds. We further show that the set of parameters for which the condition is not also necessary has Lebesgue measure zero. Given the strictness of these conditions, we introduce a notion of partial directed information, which can be used to bound the bias of directed information estimates when the directed information recipient is not itself Markov. Lastly we estimate this bound on simulations in a variety of settings to assess the extent to which the bias should be cause for concern.
\end{abstract}



\IEEEpeerreviewmaketitle


\section{Introduction}
The directed information (DI) is a popular measure of asymmetric relationships between two stochastic processes. Since its origination in 1973 \cite{marko1973bidirectional} and its reemergence in 1990 \cite{massey1990causality}, the DI has been increasingly pervasive throughout science and engineering disciplines. When using the DI to study the inter-process relationships exhibited by real data, i.e. when the true underlying joint statistics are unknown, it is necessary to utilize DI estimation techniques. DI estimators have been studied extensively in the literature using a variety of approaches, including sequential estimation using universal probability assignments \cite{jiao2013universal}, maximum likelihood estimation of generalized linear models for DI between point processes \cite{quinn2011estimating}, $k$-NN estimation \cite{murin2017k}, and plug-in estimation \cite{kontoyiannis2016estimating}. With a couple exceptions, when estimating the DI from $Y$ to $X$, these estimators assume that (i) $X$ and $Y$ are jointly stationary ergodic Markov processes and (ii) $X$ is itself a jointly stationary ergodic Markov process of the same order. While \cite{jiao2013universal} includes theoretical results for the non-Markov setting, only the context tree weighting (CTW) based estimators (which assume (i) and (ii)) are implemented due to the computational complexity of universal probability assignments for general finite-alphabet stationary ergodic sequences. In \cite{kontoyiannis2016estimating} it is noted that when assumption (ii) does not hold, the quantity being estimated is in fact not the DI, but rather an upper bound for the DI. Despite the common adoption of assumptions (i) and (ii), the conditions under which they hold and the implications when they do not are not well studied. Our present work seeks to fill this gap in order to ensure that the estimation of DI across scientific disciplines can be conducted in a manner such that the results are reliable.

Relevant discussions regarding the issues surrounding assumption (ii) have been held in the literature on Granger causality (GC) \cite{granger1969investigating}. GC can be viewed as a special case of DI where the processes in question obey a vector autoregressive (VAR) model with Gaussian noise. It is noted in the GC literature that subsets of finite-order VAR processes are in general infinite order autoregressive processes \cite{stokes2017study}. Thus, estimating a ``restricted'' model (i.e. one where the candidate influencer is hidden) from data requires estimating a truncated model and induces a bias-variance trade-off. For the linear Gaussian case, this issue can be avoided by computing the restricted model directly from the full model using the Yule-Walker equations \cite{barnett2014mvgc}. Unfortunately, there is no clear extension of this approach for arbitrary Markov processes, and other techniques are required.

We here employ a Bayesian network perspective to identify when the independence statements required by DI estimators hold. In particular, by representing a collection of interacting processes as a Bayesian network, we can use the d-separation criterion to identify conditional independencies in relevant subsets of the network. Using this perspective, we provide sufficient conditions under which assumptions (i) and (ii) are satisfied and show that these conditions are also necessary with the exception of a set of parameters with Lebesgue measure zero. We further present a bound for the estimation bias that can be estimated reliably without requiring assumption (ii). Finally, to understand the magnitude of the biases in question, we compute the proposed bound for simulated processes in a variety of problem settings.
\section{Preliminaries}
\subsection{Notation}
We will be considering collections of jointly stationary discrete processes $X$, $Y$, and $Z$, where, at any time $i$, $X_i\in\calx$, $Y_i\in\caly$, and $Z_i\in\calz$. Without loss of generality, $Z$ may represent a collection of processes $(Z^{(1)},\dots,Z^{(m)})\in \calz_1\times\dots\times\calz_m\triangleq\calz$. Collections of samples are indicated with superscripts as $X_i^{i+k}\triangleq\{X_i,\dots,X_{i+k}\}$ and $X^n\triangleq X_1^n$. In general, capital letters will represent random entities and lower case letters will represent their realizations. When a process is Markov of order $d$ we will refer to it as $d$-Markov, unless $d=1$, in which case we will simply refer to it as Markov. We will use $p$ to represent probability distributions, with the specific distribution being made clear from context.

\subsection{Directed Information}
Consider a collection of processes $(X,Y,Z)$. Define the causally conditional DI from $Y$ to $X$ given $Z$ as:
\begin{align}
&\ccdiyxz
= \sum_{i=1}^n I(X_i;Y^i\mid X^{i-1},Z^i) \\
&= \sum_{i=1}^n H(X_i\mid X^{i-1},Z^i) - H(X_i\mid X^{i-1},Y^i,Z^i)\label{di_hdiff}
\end{align}
\noindent and the associated causally conditional DI rate (when it exists) as:
\begin{equation}\label{dirate}
\ccdiryxz=
\lim_{n\ra\infty}\frac{1}{n}I(Y^n\ra X^n\cc Z^n).
\end{equation}
\noindent In the context of a collection of processes, the aforementioned assumptions are: (i) $(X,Y,Z)$ are jointly $d$-Markov, i.e. the second entropy term in \eqref{di_hdiff} can be simplified to $H(X_i\mid X^{i-1}_{i-d},Y^i_{i-d},Z^i_{i-d})$ and (ii) $X$ is ``conditionally $d$-Markov given $Z$'', i.e. the first entropy term can be simplified as $H(X_i\mid X^{i-1}_{i-d},Z^i_{i-d})$. Once these assumptions are made, it is clear that the DI can be estimated from data by splitting a stream $(X^n,Y^n,Z^n)$ into a collection of samples $\{(X^i_{i-d},Y^i_{i-d},Z^i_{i-d})\}_{i=d}^n$ and estimating the appropriate distributions using the methods of \cite{murin2017k,kontoyiannis2016estimating,quinn2011estimating,jiao2013universal}. The goal of this work is to understand when we can expect both of these assumptions to hold, and to understand what the consequences are of assuming they both hold when in fact only the first holds. It should be noted that while we only consider networks of processes and the causally conditional DI as above, all of the results hold when $Z=\emptyset$, in which case the standard DI is recovered and the assumptions above revert to the assumptions discussed in the introduction.

\subsection{Bayesian Networks}
To understand the conditions under which the desired independence relationships hold, we can use Bayesian networks, which represent conditional independencies in collections of random variables using a directed acyclic graph (DAG) $G=(V,E)$, where $V=\{V^{(1)},\dots,V^{(m)}\}$ is a set of random variables (equivalently nodes or vertices) and $E\subset V\times V$ is a set of directed edges that does not contain any cycles \cite{spirtes2000causation}. The parent set of a node $V^{(i)}$ in a DAG is defined as the set of nodes with arrows going into $V^{(i)}$, $\mc{P}_{V^{(i)}} \triangleq \{V^{(j)}:(V^{(j)} \ra V^{(i)})\in E\}$. The defining characteristic of a Bayesian network representation of a joint distribution over the nodes $V\sim p$ is the ability to factorize the distribution as:
\begin{equation}\label{factorization}
p(V) = \prod_{i=1}^m p(V_i \mid \mc{P}_{V^{(i)}}).
\end{equation}
\noindent If this factorization holds for a given $p$ and $G$, we say $G$ is a Bayesian network for $p$. A key concept when working with Bayesian networks is the d-separation criterion, which is used to identify subsets of nodes whose conditional independence is implied by the graphical structure. In particular, when given three disjoint subsets of nodes $A,B,C\subset V$ in a graph $G$, a straightforward algorithm (shown in Algorithm \ref{alg:dsep}) can be used to determine if $C$ d-separates $A$ and $B$. When $C$ d-separates $A$ and $B$, then for any joint distribution $p(V)$ such that $G$ is a Bayesian network for $p$, $A$ and $B$ will be conditionally independent given $C$. While the converse is not true in general (i.e. independence does not imply d-separation), it has been shown that for specific classes of Bayesian networks, the set of parameters for which the converse does \emph{not} hold has Lebesgue measure zero \cite{spirtes2000causation,meek2013strong}. When a graph $G$ and joint distribution $p$ are such that d-separation holds if and only if conditional independence holds for all subsets of nodes, then the distribution $p$ is called ``faithful'' to $G$ \cite{spirtes2000causation}.
\begin{algorithm}[H]
\caption{d-Separation \cite{lauritzen1990independence}} \label{alg:dsep}
\hspace*{\algorithmicindent} \textbf{Input}: DAG $G=(V,E)$ and disjoint sets $A,B,C\subset V$
\begin{algorithmic}[1]
\State Create a subgraph containing only nodes in $A$, $B$, or $C$ or with a directed path to $A$, $B$, or $C$
\State Connect with an undirected edge any two variables that share a common child
\State For each $c\in C$, remove $c$ and any edge connected to $c$ 
\State Make every edge an undirected edge
\State Conclude that $A$ and $B$ are d-separated by $C$ if and only if there is no path connecting $A$ and $B$
\end{algorithmic}
\end{algorithm}

\section{Characterization of Processes with Conditional Markovicity}
\subsection{Network Representation of Markov Processes}
A Bayesian network is a very natural representation for collections of Markov processes. In particular, using the chain rule to factorize the joint distribution over $n$ time steps of the processes $(X,Y,Z)$ yields:
\begin{equation}
p(X^n,Y^n,Z^n) 
=\prod_{i=1}^n
p(X_i,Y_i,Z_i \mid X^{i-1}_{i-d},Y^{i-1}_{i-d},Z^{i-1}_{i-d})\label{factorized}.
\end{equation}
\noindent We next make the additional assumption {\bf(A1)} that $X_i$, $Y_i$, and $Z_i$ are pairwise conditionally independent given the past $\{X^{i-1}_{i-d},Y^{i-1}_{i-d},Z^{i-1}_{i-d}\}$. This assumption facilitates construction of a Bayesian network, as we can rely on the arrow of time to determine the direction of arrows in the network. In the absence of {\bf(A1)}, we cannot construct a \emph{unique} Bayesian network representation of Markov processes without making alternative assumptions. This is similar reasoning to that of \cite{quinn2015directed}, where {\bf(A1)} is used for establishing the equivalence between DI graphs and minimal generative model graphs. Under {\bf(A1)}, we can further simplify \eqref{factorized} as:
\begin{equation}
p(X^n,Y^n,Z^n) 
=\prod_{i=1}^n \prod_{S\in\{X_i,Y_i,Z_i\}} 
p(S \mid X^{i-1}_{i-d},Y^{i-1}_{i-d},Z^{i-1}_{i-d})\label{time_network}.
\end{equation}
\noindent Comparing \eqref{factorization} and \eqref{time_network}, it is clear that we can represent a collection of processes as a Bayesian network by letting each node be a single time point of a process (i.e. $X_i$, $Y_i$, or $Z_i$) with parents $\mc{P}_{X_i},\mc{P}_{Y_i},\mc{P}_{Z_i}\subseteq\{X^{i-1}_{i-d},Y^{i-1}_{i-d},Z^{i-1}_{i-d}\}$. In general, there may be multiple valid Bayesian networks for a particular distribution. In this case, we note that $X_i$, $Y_i$, and $Z_i$ may not all depend on the entire set $\{X^{i-1}_{i-d},Y^{i-1}_{i-d},Z^{i-1}_{i-d}\}$. Thus, we construct a unique Bayesian network for $(X,Y,Z)$ by including an edge $S_{i-k}\ra S'_{i}$ for $S,S'\in \{X,Y,Z\}$ and $k=1,\dots,d$ if and only if:
\begin{equation}\label{construction}
I(S_{i-k};S'_i\mid\{X^{i-1}_{i-d},Y^{i-1}_{i-d},Z^{i-1}_{i-d}\}\setminus S_{i-k})>0.
\end{equation}

\subsection{Necessary and Sufficient Conditions for d-Separation}
Using the Bayesian network construction given by \eqref{construction}, we can leverage the d-separation criterion to gain a better understanding of the types of conditions which give rise to the conditional independence relationships needed for DI estimation. To start, we identify necessary and sufficient conditions for which $X_i$ will be d-separated from $(X^{i-l-1},Z^{i-l-1})$ by $(X^{i-1}_{i-l},Z^{i-1}_{i-l})$. In other words, the following theorem gives us a characterization of processes that are guaranteed to have the conditional independence relationships typically assumed by DI estimators:
\begin{theorem} \label{thm:dsep}
Let $(X,Y,Z)$ be a collection of jointly stationary $d$-Markov processes satisfying {\bf(A1)}. If $\ccdiyxz=0$, then $X$ is conditionally $d$-Markov given $Z$. If $\ccdiyxz>0$, $X$ is conditionally Markov given $Z$ of order $2d$ or less if:
\begin{equation}\label{sufficient}
I(Y_j;Y_k\mid X^i,Z^i)=0 \ \forall j < k \le i
\end{equation}
If $\ccdiyxz>0$ but \eqref{sufficient} is not satisfied, there will not exist any positive integer $l$ such that $(X_{i-l}^{i-1},Z_{i-l}^{i-1})$ d-separates $X_i$ from $(X^{i-l-1},Z^{i-l-1})$ in the Bayesian network generated according to \eqref{construction}.
\end{theorem}

\begin{proof}
The first statement of the theorem follows trivially from the removal of $Y^{i-1}_{i-d}$ from $p(X_i \mid X^{i-1}_{i-d},Y^{i-1}_{i-d},Z^{i-1}_{i-d})$. Now assume that \eqref{sufficient} holds. Note that:
\begin{align}
&p(X_i \mid X^{i-1},Z^{i-1})\nonumber\\
&=\sum_{y_{i-d}^{i-1}} 
p(X_i \mid X^{i-1},y_{i-d}^{i-1},Z^{i-1})
\prod_{j=i-d}^{i-1}p(y_j \mid X^{i-1},Z^{i-1}) \label{indep} \\
&=\sum_{y_{i-d}^{i-1}} 
p(X_i \mid X^{i-1}_{i-d},y_{i-d}^{i-1},Z^{i-1}_{i-d})
\prod_{j=i-d}^{i-1}p(y_j \mid X_{j-d}^{i-1},Z^{i-1}_{j-d}) \label{jmarkov} \\
&=\sum_{y_{i-d}^{i-1}} 
p(X_i \mid X^{i-1}_{i-2d},y_{i-d}^{i-1},Z^{i-1}_{i-2d})
\prod_{j=i-d}^{i-1}p(y_j \mid X_{i-2d}^{i-1},Z^{i-1}_{i-2d}) \label{addish_condish}\\
&= p(X_i \mid X_{i-2d}^{i-1},Z^{i-1}_{i-2d}) \nonumber
\end{align}

\noindent where \eqref{indep} follows from the chain rule and the conditional independence of $y_{i-d}^{i-1}$ given $(X^{i-1},Z^{i-1})$, \eqref{jmarkov} follows from the joint Markovicity of $X$ and $Y$ and the conditional independence of $y_{i-d}^{i-1}$, and \eqref{addish_condish} follows from the conditional independence of the past and the future given the present for Markov processes. Thus it follows that $X$ is conditionally Markov given $Z$ of order at most $2d$.

Now assume $\ccdiyxz>0$ but \eqref{sufficient} does not hold. Then we will show there is no positive integer $l$ such that $(X_{i-l}^{i-1},Z_{i-l}^{i-1})$ d-separates $(X^{i-l-1},Z^{i-l-1})$ from $X_i$. To do this, we first note that $(X^i,Z^i)$ does not d-separate $Y_j$ and $Y_k$, because if it did, they would be conditionally independent. As such, when performing the d-separation algorithm given by Algorithm \ref{alg:dsep}, $Y_j$ and $Y_k$ will be connected by an undirected edge after completing step 4. Furthermore, if we let $\tau_1 = k-j$, then by the joint stationarity of $(X,Y,Z)$, \emph{every} $Y_i$ will be connected to $Y_{i-\tau_1}$ at the end of step 4. Furthermore, we know that $I(Y^n\ra X^n\mid\mid Z^n)>0$ implies that for some $q\le m$, there is a directed edge from $Y_q$ to $X_m$. Letting $\tau_2 = m-q$, we know from the joint stationarity of $(X,Y,Z)$ that for every $X_i$, there is an incoming directed edge from $Y_{i-\tau_2}$. As such, at the end of step 4, every $X_i$ will be part of an undirected path connecting $Y_{i-\tau_2}$, $Y_{i-\tau_2-\tau_1}$, $Y_{i-\tau_2-2\tau_1}$, $\dots$. Thus, for any $l\ge 1$ this path can be followed $r$ steps such that $r\tau_1>d$. Then we know that $Y_{i-\tau_2-r\tau_1}$ is connected via an undirected edge to $X_{i-\tau_2-r\tau_1+\tau_2}=X_{i-r\tau_1}$. Recalling that in step 3 of the d-separation algorithm, $(X_{i-l}^{i-1},Z_{i-l}^{i-1})$ have been removed from the graph, we note that since $i-r\tau_1<i-l$, $X_{i-r\tau_1}$ is in the graph. Thus, there is an undirected path connecting $X_{r\tau_1} \in X^{i-l-1}$ and $X_i$, which implies that $(X_{i-l}^{i-1},Z_{i-l}^{i-1})$ does not d-separate $(X^{i-l-1},Z^{i-l-1})$ and $X_i$ for any $l$.
\end{proof}
We can see that the conditions presented by Theorem \ref{thm:dsep} are rather restrictive. With regard to the processes for which we cannot \emph{guarantee} the desired conditional independence relations (i.e. those not satisfying \eqref{sufficient}), the only distributions for which the assumptions in question hold are those that are \emph{unfaithful} to their graphs. While there is ample discussion in the literature noting that these distributions are typically not seen in practice (see \cite{spirtes2000causation} and citations therein), a formal characterization within the present context is desired.

\subsection{Completeness of d-Separation}
For a DAG $G=(V,E)$, define $\Gamma_G\subset\R^M$ to represent the set of $M$ parameters needed to specify \emph{all} discrete distributions $p(V)$ such that the $G$ is a Bayesian network for $p$. Further define $\Gamma_G^u\subset\Gamma_G$ to be the subset of those distributions that are unfaithful to $G$. Then, it was shown in \cite{meek2013strong} the $\Gamma_G^u$ has Lebesgue measure zero with respect to $\R^M$. Unfortunately, this result cannot be directly applied to our problem. Let $\Theta_G\subset\R^N$ represent the set of parameters defining discrete jointly stationary $d$-Markov processes satisfying {\bf(A1)} for which $G$ gives the Bayesian network constructed according \eqref{construction}. Defining the probabilities $\theta^{s_i}_{x,y,z}\triangleq p(s_i\mid x_{i-d}^{i-1},y_{i-d}^{i-1},z_{i-d}^{i-1})$ for $s\in\{x,y,z\}$, we can see that $N\triangleq (|\calx|+|\caly|+|\calz|-3)|\calx|^d|\caly|^d|\calz|^d$ many of these parameters uniquely define such a process. For a particular process, the collection of all these parameters is given by $\theta\in\Theta_G \subset \mathbb{R}^N$. Next define $\Theta_G^u\subset\Theta_G$ to be the subset of parameterizations such that the distribution $p$ induced by $\theta \in \Theta_G^u$ is unfaithful to $G$. It is clear that, due to the stationarity constraint, $N<<M$, and the Lebesgue measure of $\Gamma^u_G$ with respect to $\R^M$ does not tell us what the Lebesgue measure of $\Theta^u_G$ is with respect to $\R^N$. We seek to know when we can expect $X$ to be conditionally $d$-Markov given $Z$ despite the conditional independence not being implied by d-separation, i.e. when $p(X^n,Y^n,Z^n)$ is unfaithful. Using a similar technique to \cite{meek2013strong}, the following theorem states that, when $d=1$, the set of such parameters has Lebesgue measure zero:

\begin{theorem}\label{thm:necessary}
The set of parameters defining a collection $(X,Y,Z)$ of jointly stationary irreducible aperiodic Markov processes such that there exists a positive integer $l$ where $X$ is conditionally $l$-Markov given $Z$ but $(X_{i-l}^{i-1},Z_{i-l}^{i-1})$ does not d-separate $X_i$ from $(X^{i-l-1},Z^{i-l-1})$ in the Bayesian network constructed by \eqref{construction} has Lebesgue measure zero with respect to $\R^N$.
\end{theorem}
\begin{proof}
We will show that the statement holds for a fixed $l$, noting that a countably infinite union of measure zero sets has measure zero. First note that, if $X$ is conditionally $l$-Markov given $Z$, then for any $x_{i-l-1}^{i-1}\in\calx^l, x'_{i-l-1}\in\calx$, $z_{i-l-1}^{i-1}\in\calz^l, z'_{i-l-1}\in\calz$, the following equality must hold:
\begin{equation}\label{constraint}
    p(x_i \mid x_{i-l-1}^{i-1},z_{i-l-1}^{i-1})=p(x_i \mid \tilde{x}_{i-l-1}^{i-1},\tilde{z}_{i-l-1}^{i-1})
\end{equation}
\noindent where we define $\tilde{x}_{i-l-1}^{i-1}\triangleq \{x_{i-l}^{i-1},x'_{i-l-1}\}$ and $\tilde{z}_{i-l-1}^{i-1}\triangleq \{z_{i-l}^{i-1},z'_{i-l-1}\}$. We will demonstrate that the equation given by \eqref{constraint} amounts to solving a polynomial function of the parameters $\theta$. It is shown in \cite{okamoto1973distinctness} that the set of solutions to a non-trivial polynomial (i.e. one that is not solved by all of $\R^N$) will have Lebesgue measure zero with respect to $\R^N$. Focusing on the left side of \eqref{constraint}, we see that:
\begin{align}
&p(x_i \mid x_{i-l-1}^{i-1},z_{i-l-1}^{i-1}) =\sum_{\ym{l-1}^{i-1}} \tabp{x_i}{x,y,z}p(\ym{l-1}^{i-1}\mid \xm{l-1}^{i-1},\zm{l-1}^{i-1}) \nonumber \\
&=\sum_{\ym{l-1}^{i-1}} \tabp{x_i}{x,y,z}\frac{p(\xm{l-1}^{i-1},\ym{l-1}^{i-1}, \zm{l-1}^{i-1})}
{p( \xm{l-1}^{i-1},\zm{l-1}^{i-1})} \nonumber \\
&=\frac
{\sum_{\ym{l-1}^{i-1}} \tabp{x_i}{x,y,z}\pi(\xm{l-1},\ym{l-1}, \zm{l-1})
\prod_{j=1}^{l}\tabp{(x,y,z)_{i-j}}{x,y,z}}
{\sum_{\tilde{y}_{i-l-1}^{i-1}}\pi(\xm{l-1},\tilde{y}_{i-l-1}, \zm{l-1})
\prod_{j=1}^{l}\tabp{(x,\tilde{y},z)_{i-j}}{x,\tilde{y},z}}\label{hipi}
\end{align}
\noindent where $\pi:\calx\times\caly\times\calz\ra [0,1]$ is the invariant distribution and $\tabp{(x,y,z)_{i}}{x,y,z}\triangleq\tabp{x_i}{x,y,z}\tabp{y_i}{x,y,z}\tabp{z_i}{x,y,z}$. Next, define a matrix $A\in\R^{|\calx||\caly||\calz|\times|\calx||\caly||\calz|}$ containing the transition probabilities, i.e. $A_{j,k}=\tabp{R_k}{R_j}$ some enumeration $R$ over the $|\calx||\caly||\calz|$ possible values taken by $(X,Y,Z)$. Then we can represent $\pi$ in vector form $\vec{\pi}\in[0,1]^{|\calx||\caly||\calz|}$ as a solution to $\vec{\pi}=\vec{\pi} A$. Given $(A^T-I)\vec{\pi}=0$, it is straightforward to show that each element of $\vec{\pi}_j$ (and thus each value of $\pi(x,y,z)$) can be written as fractions of polynomial functions of the entries of $A$, each of which is one of the parameters in $\theta$. As such, \eqref{hipi} can be written using fractions of polynomial functions of $\theta$. Repeating this process, we can see that the same applies to the RHS of \eqref{constraint}. Thus, we can represent \eqref{constraint} as a polynomial function of $\theta$ by recursively multiplying both sides by any term that appears in the denominator on either side. Finally, we note that the polynomial given by \eqref{constraint} is trivial only if \emph{every} process is a solution. Though omitted here for brevity, it can be show that the polynomial is non-trivial by constructing a counterexample.\end{proof}
\noindent It should be noted that the challenge for situations where $d>1$ arises in the representation of the invariant distribution as the solution to a matrix vector multiplication, and thus other proof techniques may be required.

\section{Quantifying Estimation Bias}
We have shown that DI estimators are reliant upon a condition that is unlikely to be satisfied. Thus, we now define two augmented notions of DI that do not require $X$ to be conditionally Markov in order to be accurately estimated.
\begin{definition}
The $k^{th}$-order causally conditional truncated directed information (TDI) from $Y$ to $X$ given $Z$ is defined as:
\begin{equation}
    I_T^{(k)}(Y^n \ra X^n \cc Z^n)
    \triangleq
    \sum_{i=1}^n I(X_i;Y_{i-k}^i \mid X^{i-1}_{i-k},Z^{i}_{i-k})
\end{equation}
\end{definition}
\noindent The TDI in its unconditional form is discussed in \cite{kontoyiannis2016estimating} in the context of plug-in estimators of DI. Should both Markovicity and conditional Markovicity hold for a collection of processes, then the TDI and the DI are equivalent. However, having shown that conditional Markovicity is unlikely to hold, we here name the TDI to emphasize that it is a fundamentally different measure from the traditional DI.
\begin{definition}
The $k^{th}$-order causally conditional partial directed information (PDI) from $Y$ to $X$ given $Z$ is defined as:
\begin{equation}
    I_P^{(k)}(Y^n \ra X^n \cc Z^n)
    \triangleq
    \sum_{i=1}^n I(X_i;Y_{i-k}^i \mid X^{i-1},Y^{i-k-1},Z^{i})
\end{equation}
\end{definition}
\noindent The PDI can be thought of as measuring the unique influence of the $k$ most recent samples of $Y$ on $X$. It is important to note that, under the assumption that $(X,Y,Z)$ are jointly $d$-Markov, we have that:
\begin{align*}
&I(X_i;Y_{i-k}^i\mid X^{i-1},Y^{i-k-1},Z^i)=\\
&H(X_i\mid X^{i-1}_{i-k-d},Y^{i-k-1}_{i-k-d},Z^i_{i-k-d})- H(X_i\mid X^{i-1}_{i-d},Y^{i}_{i-d},Z^i_{i-d})
\end{align*}

\noindent Thus, it is clear that estimators of DI can be extended to estimate the PDI without the additional requirement of conditional Markovicity, though the details of these estimators are postponed for future work. Defining the TDI and PDI rates $\bar{I}_T^{(k)}$ and $\bar{I}_P^{(k)}$ to be the normalized limits analogous with the DI rate given by \eqref{dirate}, we are able to bound the DI rate from above and below as follows:

\begin{theorem}\label{thm:dibounds}
Let $(X,Y,Z)$ be jointly stationary $d$-Markov. For $k_1\ge 1$ and $k_2\ge d$, the  causally conditional PDI and TDI rates bound the DI rate as:
\begin{equation}
\bar{I}_P^{(k_1)}(Y\ra X \cc Z) \le \ccdiryxz \le \bar{I}_T^{(k_2)}(Y\ra X \cc Z)
\end{equation}
with both bounds becoming equalities as $k_1,k_2\ra\infty$.
\end{theorem}
\begin{proof}
Note that for any $k_1\ge 1$ and $k_2\ge d$:
\begin{align}
&H(X_i\mid X^{i-1},Y^{i-k_1-1},Z^i) -
H(X_i\mid X^{i-1},Y^{i},Z^i) \label{pdiproof}\\
&\le H(X_i\mid X^{i-1},Z^i) -H(X_i\mid X^{i-1},Y^{i},Z^i)\label{reduce1} \\
&\le H(X_i\mid X^{i-1}_{i-k_2},Z^i_{i-k_2}) -H(X_i\mid X^{i-1},Y^{i},Z^i)\label{reduce2} \\
&= H(X_i\mid X^{i-1}_{i-k_2},Z^i_{i-k_2}) -H(X_i\mid X^{i-1}_{i-d},Y^{i}_{i-d},Z^i_{i-d}) \label{d_markov}\\
&\le H(X_i\mid X^{i-1}_{i-k_2},Z^i_{i-k_2}) -H(X_i\mid X^{i-1}_{i-k_2},Y^{i}_{i-k_2},Z^i_{i-k_2}) \label{reduce3}
\end{align}
\noindent where \eqref{reduce1}, \eqref{reduce2}, and \eqref{reduce3} follow from conditioning reduces entropy and \eqref{d_markov} follows from joint $d$-Markovicity of $(X,Y,Z)$. Taking the sum over $i=1,\dots,n$ and the normalized limit as $n\ra\infty$ gives the desired result, noting that \eqref{pdiproof}, \eqref{reduce1}, and \eqref{reduce3} become the PDI, DI, and TDI rates, respectively.
\end{proof}
\section{Simulations}
In the above sections we have demonstrated that while one cannot reasonably expect data to satisfy the necessary assumptions for obtaining unbiased estimates of DI, the TDI and PDI can be used to provide upper and lower bounds for the true DI. A natural next question is, how significant is the difference between PDI and TDI? To address this question, we simulate a pair of jointly stationary Markov discrete processes in four settings, each characterized by a particular simplification of the generative distribution $p(X_i,Y_i\mid X^{i-1},Y^{i-1})$:
\begin{align*}
    & \ p(X_i\mid Y_{i-1})p(Y_i\mid Y_{i-1}) \tag{S1}\\
    & \ p(X_i\mid X_{i-1},Y_{i-1})p(Y_i\mid Y_{i-1})  \tag{S2} \\
    & \ p(X_i\mid X_{i-1},Y_{i-1})p(Y_i\mid X_{i-1},Y_{i-1})  \tag{S3} \\
    & \ p(X_i\mid X^{i-1}_{i-2},Y^{i-1}_{i-2})p(Y_i\mid X^{i-1}_{i-2},Y^{i-1}_{i-2}) \tag{S4} 
\end{align*}
\noindent For each of these graphical structures, we conducted 100 experiments with $|\calx|=|\caly|=4$ for (S1)-(S3) and $|\calx|=|\caly|=3$ for (S4). In each experiment, the parameters were sampled as independent exponential random variables and then appropriately normalized, yielding parameters drawn uniformly from the probability simplex \cite{devroye1986non}. Using the sampled parameters, sequences $(x^n,y^n)$ were generated with $n=300000$ (large enough to ensure that accurate estimates of the TDI and PDI could be obtained). $\bar{I}_T^{(k)}(Y\ra X)$ and $\bar{I}_P^{(k)}(Y\ra X)$ were estimated using CTW estimators in the style of $\hat{I}_3$ in \cite{jiao2013universal} for $k=d$, $d+1$, and $d+2$\footnote{Code and additional figures can be found in the following repository: \url{https://github.com/gabeschamberg/directed_info_bias}.}.
\begin{figure}[b]
  \includegraphics[keepaspectratio, width=\linewidth]{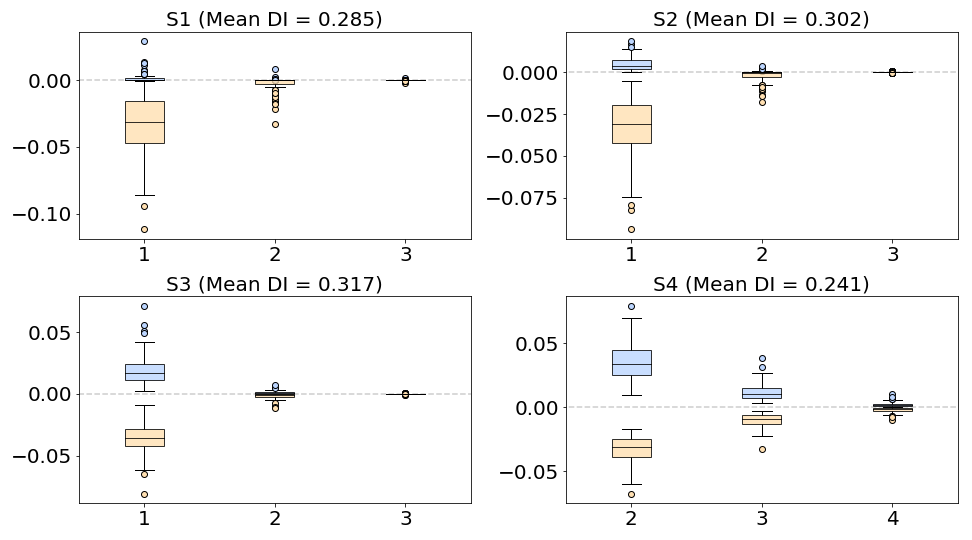}
  \caption{Difference between TDI and DI (blue) and PDI and DI (orange) for different values of $k$ (x-axis) under different process structures (panels).} 
  \label{fig:boxplot} 
\end{figure}
Figure \ref{fig:boxplot} shows boxplots representing $\hat{\bar{I}}_T^{(k)}(Y\ra X)-\hat{\bar{I}}(Y\ra X)$ and $\hat{\bar{I}}_P^{(k)}(Y\ra X)-\hat{\bar{I}}(Y\ra X)$ for varying values of $k$ along with the mean (across trials) DI rate, which was determined by the value converged upon by the TDI and PDI. We can see that the TDI is very close to the true DI for simpler structures (i.e. (S1) and (S2)), and in these cases the PDI is not a very tight lower bound. However, for the fully connected structures (S3) and (S4) the TDI may be considerably larger than the true DI and the PDI serves as a useful lower bound for the true DI. This figure suggests that while (S4) is not covered by Theorem \ref{thm:necessary}, alternative proof techniques may exist for demonstrating that the results hold for $d>1$.



\printbibliography

\end{document}